\documentclass[aps,amsmath,amssymb,11pt,nofootinbib]{revtex4}
\usepackage{bm,amsthm,epsfig}

\newtheorem{lem}{Lemma}
\newtheorem{thm}{Theorem}

\newtheorem*{conj}{Conjecture}

\newcommand{\qG}{\mathcal{G}}
\newcommand{\tqG}{\widetilde{\mathcal{G}}}
\newcommand{\bqG}{\quat{\qG}}
\newcommand{\bqR}{\quat{\mathcal{R}}}
\newcommand{\quat}[1]{\bm{#1}}
\newcommand{\tbR}{\widetilde{\quat{R}}}
\newcommand{\tbG}{\widetilde{\quat{G}}}
\newcommand{\bbE}{\mathbb{E}\,}
\newcommand{\bq}{\quat{q}}
\newcommand{\bs}{\quat{s}}
\newcommand{\bK}{\quat{K}}
\newcommand{\tbK}{\widetilde{\quat{K}}}
\newcommand{\bt}{\quat{t}}
\newcommand{\bj}{\quat{j}}
\newcommand{\bZ}{\quat{Z}}
\newcommand{\bA}{\quat{A}}
\newcommand{\bB}{\quat{B}}
\newcommand{\bD}{\quat{D}}
\newcommand{\bX}{\quat{X}}
\newcommand{\bY}{\quat{Y}}
\newcommand{\bR}{\quat{R}}
\newcommand{\bx}{\quat{x}}
\newcommand{\by}{\quat{y}}
\newcommand{\bG}{\quat{G}}
\newcommand{\bE}{\quat{E}}
\newcommand{\sF}{\mathcal{F}}
\newcommand{\bmat}{\left(\begin{array}{cc}}
\newcommand{\emat}{\end{array}\right)}

\begin{document}

\title{\begin{Large}Universal sum and product rules for random matrices \end{Large}}
\author{Tim Rogers\footnote{Electronic address: timothy.c.rogers@kcl.ac.uk}}\affiliation{Department of Mathematics, King's College London, Strand, London WC2R 2LS, United Kingdom}

\begin{abstract}
The spectral density of random matrices is studied through a quaternionic generalisation of the Green's function, which precisely describes the mean spectral density of a given matrix under a particular type of random perturbation. Exact and universal expressions are found in the high-dimension limit for the quaternionic Green's functions of random matrices with independent entries when summed or multiplied with deterministic matrices. From these, the limiting spectral density can be accurately predicted.
\end{abstract}
\maketitle
\section{Introduction}
A central problem in random matrix theory (RMT) is, and always has been, to determine the distribution of eigenvalues of a random matrix ensemble. For a given $N\times N$ matrix $X$ the statistics of the eigenvalues $\{\lambda_i^{(X)}\}_{i=1}^N$ are captured in the spectral density
\begin{equation*}
\varrho(\lambda;X)=\frac{1}{N}\sum_i\delta\left(\lambda_i^{(X)}-\lambda\right)\,,
\end{equation*}
a generalised function of the complex variable $\lambda$. Of particular interest are those ensembles for which the spectral density converges to a non-random limit as $N\to\infty$, even more so when this limit is universal in the sense that it is independent of the distributions of the entries of the matrix. For example, in the case that matrices $\{A_N\}$ are taken to be Hermitian with otherwise independent entries satisfying a Lindeberg-type condition, the so-called `simple approach' of Pastur and collaborators \cite{Pastur1999, Khorunzhy1993, Khorunzhy1996, Khorunzhy2001} provides a robust and straightforward framework to prove the convergence of $\varrho(\lambda;A_N+D_N)$ for suitable deterministic Hermitian matrices $\{D_N\}$. The limiting spectral density in this case is given by the celebrated Pastur equation \cite{Pastur1972}. \par
Unfortunately, it has long been recognised that many of the techniques used in the analysis of Hermitian random matrices simply do not apply to their non-Hermitian counterparts (see, for example, comments in \cite{Bai1997} and \cite{Pan2007}), and universal results like the Pastur equation are notoriously hard to come by. A case in point is the famous Circular law which, after many years of research and several partial results \cite{Bai1997,Pan2007,Girko1985,Girko2004,Gotze2007,Gotze2007a}, has only recently been proved in full generality \cite{Tao2007}. \par
The purpose of this paper is to show how the methods of the simple approach can be re-engineered for use in the general setting of non-Hermitian matrices. In doing so, we will prove an extension of the Pastur equation with no Hermiticity requirements placed on either the random or deterministic part. A similar result is found for the product of a random and a deterministic matrix, a result which (for obvious reasons) has no analogue in Hermitian RMT. \par 
Amongst the numerous physical applications of non-Hermitian random matrix theory, we mention only a pedagogical example of quantum chromodynamics (QCD) with finite chemical potential. The density of states, derived first by Stephanov in \cite{Stephanov1996}, can be recovered easily as an example of the sum rule.\par
In common with many techniques in the study of Hermitian random matrices, the simple approach relies upon the well-behaved nature of the Green's function of Hermitian matrices, defined by
\begin{equation*}
G(\lambda;X)=\frac{1}{N}\sum_i\left(X-\lambda\right)^{-1}_{ii}\,.
\end{equation*}
For $\lambda$ away from the real axis (to which the eigenvalues of Hermitian matrices are confined) the Green's function is analytic and its imaginary part gives a smooth and $N$-independent regularisation of the spectral density. If $X$ is non-Hermitian however, the eigenvalues invade the complex plane and the Green's function provides no such regularisation. Instead, we have the exact relation
\begin{equation}
\varrho(\lambda;X)=-\frac{1}{\pi}\frac{\partial}{\partial\overline{\lambda}}\,G(\lambda;X)\,,
\label{rhoG}
\end{equation}
where $\partial/\partial\overline{\lambda}$ is the anti-holomorphic derivative. It has been frequently suggested to study the spectral density of non-Hermitian matrices through a regularised form of the Green's function or some related object\footnote{For instance, an electrostatic potential introduced in \cite{Sommers1988} and used frequently thereafter.}. This is usually achieved by association with an Hermitian proxy, of which there are several equally good (and often equivalent) choices. In an influential series of papers \cite{Feinberg1997,Feinberg1997a,Feinberg2001} Feinberg and Zee worked with $2N\times 2N$ block matrices of the form
\begin{equation}
H=\bmat i\varepsilon & (X-\lambda) \\ (X-\lambda)^\dag & i\varepsilon \emat\,, 
\label{FZ}
\end{equation}
a process they christened `Hermitianisation'. \par Around the same time, Janik, Nowak and collaborators proposed a similar block extension technique, obtaining a generalisation of the Green's function with a quaternionic structure \cite{Janik1997,Janik1997a}. In these and subsequent works \cite{Jarosz2004,Jarosz2006}, the application of free probability theory to this quaternionic formalism has yeilded many interesting results, including for sums of unitary random matrices \cite{Jarosz2007} and infinite products of large random matrices \cite{GudowskaNowak2003}. In this paper, we will take a different direction and instead seek to adapt the techniques of the simple approach to Hermitian RMT, obtaining results that hold for matrices with independent entries of unspecified distributions. \par
Before we procede, it should be noted that there is a significant drawback to the premise of introducing an $\varepsilon>0$ regularisation to the Green's function of a non-Hermitian matrix. Simply put, if the matrices involved are not normal, there may be parts of the complex plane far from the spectrum in which the Green's function is nevertheless very large. In practical terms, this causes great difficultly in justifying the exchange of the limits $N\to\infty$ and $\varepsilon\to0$ \cite{Khoruzhenko2003}. One route around this problem involves the analysis of least singular values of the matrices involved, and has formed a large part of recent work on the Circular law \cite{Pan2007,Gotze2007,Gotze2007a,Tao2007,Chafai2007}. A tentative link is often made to the influence of the pseudospectrum, though this idea is rarely expanded upon. \par
In the present paper no attempt is made to explicitly tackle this problem, however, we are able to offer a remarkable relation between the (regularised) quaternionic Green's function and the mean spectral density of a given matrix under a particular type of random perturbation. \par
The main results of the paper are stated in Section 2, together with some brief examples of the sum and product rules and a conjecture regarding a certain class of random matrices with an asymptotically spherical spectral density. Proofs of all the main theorems are provided in Section 3 and the final section contains a discussion of the limitations of the work and some possible directions for future research. The notation used to define and manipulate the quaternionic Green's function is introduced in the remainder of this section.
\newpage
\subsection*{Notation}
We will be doubling the size of our matrices. To simplify the formulas, we introduce the following notation:  \par
Denote by $\bX$ the $2N\times 2N$ matrix composed of $N^2$ blocks of size $2\times2$ whose $i^{\rm th}, j^{\rm th}$ block is given by
\begin{equation*}
\bX_{ij}=\bmat X_{ij} & 0 \\ 0 & \overline{X_{ji}} \emat\,.
\end{equation*}
We will always use Roman indices $i,j,k,l$ to refer to the $2\times 2$ blocks of boldface matrices, rather than the individual entries. All sums over Roman indices run from 1 to $N$. To work with quaternions, we introduce the $2\times 2$ matrix 
\begin{equation*}
\bj=\bmat 0 & i \\ i & 0 \emat\,.
\end{equation*}
Now, if $q=a+bj$ is a quaternion (i.e. $a$ and $b$ are complex numbers and $j$ is a quaternionic basis element), then we have the matrix representation
\begin{equation}
\bq=\bm{a}+\bm{b}\bm{j}=\bmat a & ib \\ i\overline{b} & \overline{a} \emat\,.
\label{matrep}
\end{equation}
This is an isomorphism, and $|q|=\|\bq\|$, where $\|\cdot\|$ is the spectral norm. When $\bq$ is a $2\times 2$ matrix, and $\bX$ a $2N\times 2N$ matrix, we use the shorthands
\begin{equation*}
\bq\bX=(\bq\otimes I_N)\bX\,,\quad\textrm{and}\quad (\bX+\bq)=\bX+\bq\otimes I_N\,.
\end{equation*}
In addition to the usual operations, we define an elementwise product for quaternions
\begin{equation*}
(a+bj)\cdot(c+dj)=ac+bdj\,.
\end{equation*}
Note that the matrix representation of an elementwise product of quaternions is not the same as the usual elementwise product of the matrices, in fact we use
\begin{equation*}
\bmat a & ib \\ i\overline{b} & \overline{a} \emat\cdot\bmat c & id \\ i\overline{d} & \overline{c} \emat=\bmat ac & ibd \\ i\overline{bd} & \overline{ac} \emat\,.
\end{equation*}

\section{Main results}
Let $X$ be an $N\times N$ matrix, $\lambda$ a complex variable and $\varepsilon$ a strictly positive real number. Putting $q=\lambda+\varepsilon j$, we define the $2N\times 2N$ `resolvent'
\begin{equation*}
\bm{\mathcal{R}}(q;X)=\left(\bX-\bq\right)^{-1}\,.
\end{equation*}
To connect with other approaches, note that there exists a permutation matrix $P$ such that
\begin{equation*}
\left(\bX-\bq\right)=P \bmat 0 & I_N \\ I_N & 0\emat H P^{-1}\,,
\end{equation*}
where $H$ is the `Hermitianised' block matrix given in (\ref{FZ}). The quaternionic Green's function of $X$ at $q$ is then defined to be the quaternion $\qG(q;X)$ with matrix representation 
\begin{equation*}
\bqG(q;X)=\frac{1}{N}\sum_i \bm{\mathcal{R}}(q;X)_{ii}\,.
\end{equation*}
Without the hypercomplex part, the quaternionic Green's function agrees with the usual Green's function, $$\qG(\lambda+0j;X)=G(\lambda;X)+0j\,.$$
Adding a positive real regulariser $\varepsilon>0$, we apply (\ref{rhoG}) to obtain from $\qG$ a regularisation of the spectral density,
\begin{equation*}
\varrho_\varepsilon(\lambda;X)=-\frac{1}{\pi}\textrm{Re}\,\frac{\partial}{\partial\overline{\lambda}}\,\qG(\lambda+\varepsilon j;X)\,.
\end{equation*}
Interestingly, this regularisation is precisely the mean spectral density of $X$ under a particular type of random perturbation.
\begin{thm}[Perturbation Formula]
Let $X$ be an arbitrary $N\times N$ matrix and $\varepsilon$ a strictly positive real number. Suppose $A$ and $B$ are random $N\times N$ matrices, with independent standard complex Gaussian entries, then
\begin{equation*}
\mathbb{E}\,\varrho(\lambda;X+\varepsilon AB^{-1})=\varrho_\varepsilon(\lambda;X)\,.
\end{equation*}
\end{thm}
A short proof is presented in the next section, using a matrix generalisation of the M\"{o}bius transformation. As a corollary of Theorem 1 we see that mean spectral density of the matrices $AB^{-1}$ is, regardless of their size, the uniform distribution on the Riemann Sphere\footnote{In fact, the full jdpf of eigenvalues for matrices of this type was found recently in \cite{Krishnapur2009}.}. We are immediately prompted to ask if this result holds in the limit $N\to\infty$ for any distribution of the entries of $A_N$ and $B_N$, in an analogue of the circular law:
\begin{conj}[The Spherical Law]
Let $\{A_N\}$ and $\{B_N\}$ be sequences of matrices of independent complex random variables of zero mean and unit variance. Then the spectral densities of the matrices $A_NB_N^{-1}$ converge to the uniform density on the Riemann sphere.
\end{conj}
This phenomenon was also noticed by Forrester and Mays \cite{Forrester2009} and can be easily derived in a non-rigorous fashion using the techniques developed below, however a full proof is likely to require a more in-depth analysis.\par
The other main results of this paper concern the quaternionic Green's function itself, and specifically the ease with which universal predictions can be made about the limiting regularised spectral density of sums and products of random and deterministic matrices. \par
Suppose we are in possession of an infinite array of complex random variables $\{\xi_{ij}\}_{i,j=1}^\infty$, with joint probability space $(\Omega,\sF,\mathbb{P})$. We assume the $\xi_{ij}$ to have the following properties:
\begin{description}
\item[A1)] $\bbE\xi_{ij}=0$ for all $i,j$
\item[A2)] $\bbE|\xi_{ij}|^2=1$ for all $i,j$
\item[A3)] There exists a finite constant $C_\xi$ such that $\bbE|\xi_{ij}|^3<C_\xi$ for all $i,j$
\item[A4)] All $\xi_{ij}$ are independent, except for the covariance $\bbE\xi_{ij}\xi_{ji}=\tau$, where $\tau\in[0,1]$.
\end{description}
A normalised $N\times N$ random matrix $A_N$ can then be constructed by taking 
\begin{equation}
\big(A_N\big)_{ij}=\frac{1}{\sqrt{N}}\xi_{ij}\,.
\label{A}
\end{equation}
Introduced in \cite{Sommers1988}, the parameter $\tau$ controls the degree of Hermiticity of $A_N$. At $\tau=1$, we have that $A_N$ is an Hermitian Wigner-class matrix, whilst at $\tau=0$ the entries are completely independent. In our calculations, $\tau$ will only appear as the real part of the quaternion $t=\tau+j$.\par
In the results below, we will characterise the quaternionic Green's functions of the sum or product of such matrices with deterministic matrices $\{D_N\}$ satisfying some or all of the assumptions
\begin{description}
\item [D1)] The quaternionic Green's functions of $\{D_N\}$ converge pointwise to the limit $\qG_D$
\item [D2)] The quaternionic Green's functions of $\{D^{-1}_N\}$ converge pointwise to the limit $\qG_{D^{-1}}$
\item [D3)] There exists a constant $d\in[0,\infty)$ such that $\sup\|D_N\|<d$ and $\sup\|D_N^{-1}\|<d$.
\end{description}
The last point is a technical assumption made for the sake of simplicity and, as we note in a later example, may not be strictly necessary.
\begin{thm}[Sum Rule]
Let $\{A_N\}$ be a sequence of random matrices given by (\ref{A}) and let $\{D_N\}$ be a sequence of deterministic matrices satisfying D1. Fix a quaternion $q=\lambda+\varepsilon j$, where $\varepsilon>1$. Then
\begin{equation*}
\qG(q;D_N+A_N)\stackrel{\mathbb{P}}{\longrightarrow}\qG(q)\,,
\end{equation*}
where $\qG(q)$ satisfies
\begin{equation}
\qG(q)=\qG_D\Big(q+t\cdot\qG(q)\,\Big)\,,
\label{sum}
\end{equation}
\end{thm}
Theorem 2 is a straightforward generalisation of the Pastur equation for the sum of deterministic and random Hermitian matrices, indeed at $\tau=1$ and $\varepsilon\to0$, equation (\ref{sum}) precisely \emph{is} the Pastur equation. For the case $\tau=0$, an equivalent result has already been found using potential theory \cite{Khoruzhenko1996}.
\begin{thm}[Product Rule]
Let $\{A_N\}$ be a sequence of random matrices given by (\ref{A}) and let $\{D_N\}$ be a sequence of deterministic matrices satisfying D1-D3. Fix a quaternion $q=\lambda+\varepsilon j$, where $\varepsilon>2d$. Then
\begin{equation}
\qG(q;D_NA_N)\stackrel{\mathbb{P}}{\longrightarrow}\qG(q)=-\big(t\cdot\tqG\big)^{-1}\qG_D\Big(-q\,\big(t\cdot\tqG\big)^{-1}\Big)\,,
\label{prodtrue}
\end{equation}
where $\tqG$ satisfies
\begin{equation}
\tqG=-q^{-1}\qG_{D^{-1}}\Big(-\big(t\cdot\tqG\big)\,q^{-1}\Big)\,.
\label{prod}
\end{equation}
\end{thm}
Unlike Theorem 2, this result is not related to any in Hermitian RMT for the simple reason that the space of Hermitian matrices is not closed under multiplication. 
\subsection*{Examples}
The statement of the sum and product rules given above concerns the behaviour of the quaternionic Green's function in the limit $N\to\infty$, for a fixed regulariser $\varepsilon$, which is taken to be large. In light of Theorem 1 we are in effect computing the limiting spectral density of matrices under a large perturbation. However, as the following examples will demonstrate, accurate predictions about the spectral densities of sums and products of matrices satisfying conditions A1-A4 can be found by naively taking $\varepsilon=0$ in the final equations.\par
\textbf{The Elliptic Law:}\\
The well-known elliptic law occurs naturally. Taking either $D_N=0$ in the sum rule, or $D_N=I_N$ in the product rule, gives $\qG(q;A_N)\stackrel{\mathbb{P}}{\longrightarrow}\qG(q)$, where
\begin{equation*}
\qG(q)=-\Big(q+t\cdot\qG(q)\,\Big)^{-1}\,.
\end{equation*}
Writing $\qG(q)=\alpha+\beta j$, we send $\varepsilon\to0$ and assume that $\beta$ stays strictly positive, obtaining
\begin{equation}
(\alpha+\beta j)(\lambda+\tau\alpha+\beta j)+1=0\,.
\label{ellip}
\end{equation}
The support of the spectral density is restricted to the region allowing a solution with $\beta>0$. The hypercomplex part of (\ref{ellip}) gives
\begin{equation*}
\alpha=-\left(\frac{x}{\tau+1}+i\frac{y}{\tau-1}\right)\,,
\end{equation*}
where $\lambda=x+iy$, and the complex part gives 
\begin{equation*}
\beta=\sqrt{1-\left(\frac{x}{\tau+1}\right)^2-\left(\frac{y}{\tau-1}\right)^2}\,.
\end{equation*}
The condition $\beta>0$ determines the elliptic support, and taking an anti-holomorphic derivative yields the spectral density inside that region:
\begin{equation*}
\rho(\lambda)=\begin{cases}\frac{1}{\pi(1-\tau^2)}\quad&\textrm{when}\quad \left(\frac{x}{\tau+1}\right)^2+\left(\frac{y}{\tau-1}\right)^2<1\\0&\textrm{otherwise.}\end{cases}
\end{equation*}
\textbf{Sum Rule:}\\
Let us consider the following random matrix model for the Dirac operator in QCD with finite chemical potential, considered in \cite{Stephanov1996}, 
\begin{equation}
D=\bmat 0 & iA_N+\mu \\ iA_N^\dag +\mu & 0 \emat\,,
\label{steph}
\end{equation}
where $A_N$ is drawn from the Gaussian Unitary Ensemble and $\mu>0$ is the chemical potential. Multiplying by an imaginary unit and filling the diagonal blocks, we suppose that in the limit $N\to\infty$, the spectral density of (\ref{steph}) may be recovered from that of $A_N+iM_N$, where 
\begin{equation*}
M_N=\bmat 0 & \mu \\ \mu & 0 \emat \otimes I_{N/2}\,.
\end{equation*}
For the quaternionic Green's function, the sum rule states that $\qG(q;A_N+iM_N)\stackrel{\mathbb{P}}{\longrightarrow}\alpha+\beta j$, where equation (\ref{sum}) now reads
\begin{equation}
\alpha+\beta j = \frac{1}{2}\Big(i\mu-q-\alpha-\beta j\Big)^{-1}-\frac{1}{2}\Big(i\mu+q+\alpha+\beta j\Big)^{-1}\,.
\label{sumex}
\end{equation}
Assuming the limit $\varepsilon\to0$, we write $q=\lambda$. The support of spectral density is then given by the region allowing a solution of (\ref{sumex}) with $\beta>0$, in this case determined by the condition
\begin{equation*}
1-\frac{1}{4}x^2-\frac{1}{4}y^2\Big(1+2\big(y^2-\mu^2\big)\Big)^2\big(y^2-\mu^2\big)^{-2}+\big(y^2-\mu^2\big)>0\,.
\end{equation*}
Inside this region, one may solve for $\alpha$ and take the anti-holomorphic derivative to determine the spectral density
\begin{equation}
\rho(\lambda)=\frac{1}{4\pi}\left(\frac{y^2+\mu^2}{(y^2-\mu^2)^2}-1\right)\,.
\label{QCDrho}
\end{equation}
As expected, this is precisely the density first recovered by Stephanov in \cite{Stephanov1996}, rotated by $\pi/2$. In that work, only the mean of the spectral density was computed, however the sum rule tells us that the quaternionic Green's function converges in probability as $N\to\infty$, suggesting weak convergence in spectral density. Moreover, only Gaussian distributed random matrices were considered in \cite{Stephanov1996}, whereas the sum rule predicts the limiting density to be universal in the sense that it is independent of distribution of the entries of $A_N$. An analogue of the construction (\ref{steph}) for orthogonal and symplectic ensembles has been studied numerically in \cite{Halasz1997}, where, interestingly, the limiting densities were found to be different to those for the unitary ensemble.\par
\textbf{Product Rule:}\\
We compute the limiting spectral density for the product matrix $D_NA_N$, where the $A_N$ are given by (\ref{A}) with $\tau=0$, and $D_N$ is a diagonal matrix with entries $D_{ii}$ drawn independently from the standard Cauchy distribution\footnote{Notice that although this choice fails the technical assumption D3, the result still appears to hold.}. With this choice of $D_N$, we have the limits
\begin{equation*}
\qG(q;D_N)\to\qG_D(q)=\frac{1}{\pi}\int_{-\infty}^\infty \frac{1}{1+r^2}(r-q)^{-1}{d}r\,,
\end{equation*}
and $\qG(q;D_N^{-1})\to\qG_D(q)$, as $N\to\infty$. As before, we assume the $\varepsilon\to0$ limit, taking $q=\lambda$. Then $\tqG(q)=\widetilde{\alpha}+\widetilde{\beta} j$, where the product rule (\ref{prod}) gives
\begin{equation*}
\widetilde{\alpha}+\widetilde{\beta} j=-\frac{1}{\pi}\int_{-\infty}^\infty \frac{1}{1+r^2}\left(\frac{\lambda}{r}+\widetilde{\beta} j\right)^{-1}{d}r\,.
\end{equation*}
Performing the integral and solving for $\widetilde{\beta}$, we obtain $\widetilde{\beta}=\big(-|\lambda|+\sqrt{|\lambda|^2+4}\big)/2$. Returning to (\ref{prodtrue}), we reach
\begin{equation*}
\qG(q;D_NA_N)\stackrel{\mathbb{P}}{\longrightarrow}-\frac{1}{\pi}\int_{-\infty}^\infty \frac{1}{1+r^2}\Big(\lambda+r\widetilde{\beta} j\Big)^{-1}{d}r=\frac{2\overline{\lambda}}{|\lambda|^2+|\lambda|\sqrt{|\lambda|^2+4}}\,,
\end{equation*}
and finally an expression for the limiting spectral density,
\begin{equation*}
\rho(\lambda)=\frac{1}{\pi}\left(\frac{1}{|\lambda|^2+|\lambda|\sqrt{|\lambda|^2+4}}-\frac{1}{|\lambda|^2+|\lambda|\sqrt{|\lambda|^2+4}+4}\right)\,.
\end{equation*}
To provide an effective comparison with numerical data, we change variables to $\gamma=\ln|\lambda|$, whose distribution is given by the pdf
\begin{equation}
\nu(\gamma)=2\frac{e^\gamma\sqrt{e^{2\gamma}+4}-e^{2\gamma}}{e^\gamma\sqrt{e^{2\gamma}+4}+e^{2\gamma}+4}\,.
\label{nugamma}
\end{equation}
Figure \ref{prodfig} shows a histogram of the log-moduli of the eigenvalues of a single such random matrix of size $N=10,000$ alongside the predicted density $\nu(\gamma)$.
\begin{figure}
\begin{center}
\includegraphics[trim=0 50 0 0, width=\textwidth]{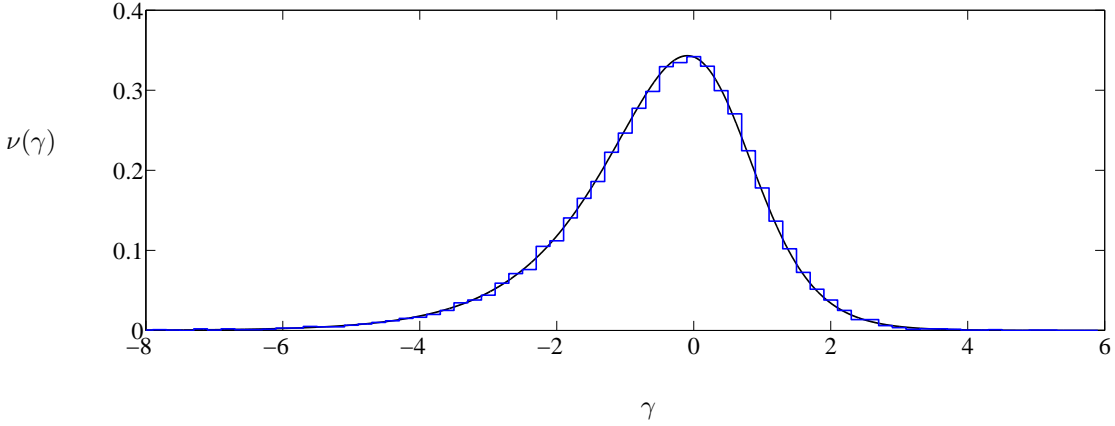}    
\put(-220,-30){$\gamma$}
\put(-460,70){$\nu(\gamma)$}
\end{center}
\caption{Blue histogram: the log-moduli of eigenvalues of a single random matrix $DA$ of size $N=10,000$, where the entries of $A$ are independent complex Gaussians with variance $1/N$, and $D$ is a diagonal matrix of Cauchy random variables. Black line: the density predicted by the product rule (Theorem 3), given in this case by equation (\ref{nugamma}) in the text.}
\label{prodfig}
\end{figure}

\section{Proof of main results}

\subsection{Proof of perturbation formula}
\begin{proof}[Proof of Theorem 1]
For the case $N=1$, the ratio of two standard complex Gaussian random variables takes the uniform density on the Riemann sphere. One proof of this fact comes from the observation that the density generated is invariant under a class of M\"{o}bius transformations. We generalise this idea to matrices. Begin by noting that, for an arbitrary $N\times N$ matrix $X$,
\begin{equation*}
\varrho_\varepsilon(\lambda;X)=-\frac{1}{N}\textrm{Tr}\frac{\partial}{\partial\overline{\lambda}}\Big((X-\lambda)^\dag (X-\lambda)+\varepsilon^2\Big)^{-1}(X-\lambda)^\dag\,,
\end{equation*}
and
\begin{equation*}
\varrho(\lambda;X+\varepsilon AB^{-1})=-\frac{1}{N}\textrm{Tr}\frac{\partial}{\partial\overline{\lambda}}\big(\varepsilon AB^{-1}+X-\lambda\big)^{-1}\,.
\end{equation*}
Theorem 1 will then follow from the stronger claim that for any matrix $Y$
\begin{equation*}
\mathbb{E}\,\big(\varepsilon AB^{-1}+Y\big)^{-1}=\big(Y^\dag Y+\varepsilon^2\big)^{-1}Y^\dag\,.
\end{equation*}
A little rearrangement leads to the equivalent statement $\mathbb{E}\,f(AB^{-1})=0$, where we have introduced the matrix M\"{o}bius transformation $f$, given by
\begin{equation*}
f(Z)=\big(Y^\dag Z-\varepsilon\big)\big(\varepsilon Z+Y\big)^{-1}\,.
\end{equation*}
Notice that if $\bmat A \\ B\emat\mapsto AB^{-1}$, then $F\bmat A \\ B\emat\mapsto f(AB^{-1})$, where $F$ is the $2N\times 2N$ block matrix
\begin{equation*}
F=\bmat Y^\dag &-\varepsilon \\ \varepsilon & Y \emat\,.
\end{equation*}
It will be useful to `normalise' $f$, introducing 
\begin{equation*}
\tilde{f}(Z)=\big(Y^\dag Y+\varepsilon^2\big)^{-1/2}\big(Y^\dag Z-\varepsilon\big)\big(\varepsilon Z+Y\big)^{-1}\big(Y Y^\dag+\varepsilon^2\big)^{1/2}\,.
\end{equation*}
Then the normalised form of $F$ is given by $\tilde{F}\in \textrm{SU}(2N)$,
\begin{equation*}
\tilde{F}=\bmat Y^\dag Y+\varepsilon^2& 0\\0 &Y Y^\dag+\varepsilon^2 \emat^{-1/2} \bmat Y^\dag &-\varepsilon \\ \varepsilon & Y \emat\,.
\end{equation*}
Now, let $A$ and $B$ be independent complex Gaussian matrices, with joint pdf
\begin{equation*}
p\bmat A \\ B\emat=\frac{1}{\pi^{2N^2}}\exp\left\{-\textrm{Tr}\left[ \bmat A^\dag & B^\dag\emat \bmat A \\ B\emat\right]\right\}\,,
\end{equation*}
and let $\bmat A' \\ B'\emat=\tilde{F}\bmat A \\ B\emat$. We perform a change of variables to find that the joint pdf of $A'$ and $B'$ is given by 
\begin{equation*}
\begin{split}
\frac{1}{\pi^{2N^2}}\exp\left\{-\textrm{Tr}\left[ \bmat A^\dag & B^\dag\emat \big(\tilde{F}^\dag \tilde{F}\big)^{-1} \bmat A \\ B\emat\right]\right\}\Big|\det \tilde{F}^{-1}\Big|=p\bmat A \\ B\emat\,.
\end{split}
\end{equation*}
Thus the density $p$ is invariant under multiplication with $\tilde{F}$ and we can conclude that the distribution of the random variable $AB^{-1}$ is invariant under $\tilde{f}$. \par
So to prove that $\mathbb{E}f(AB^{-1})=0$, it will suffice to show that $\mathbb{E}AB^{-1}=0$. But $A$ and $B$ are independent, $\mathbb{E}A=0$, and it can be shown (see \cite{Edelman1988}) that $\mathbb{E}\|B^{-1}\|=\sqrt{N\pi}<\infty$, so we are done.
\end{proof}
\subsection{Preliminaries for working with the quaternionic Green's function}
For the proofs of Theorems 2 and 3, a number of standard tools will be of repeated use.
\begin{description}
\item[The resolvent identity.]
For $N\times N$ matrices $A$ and $B$, and quaternion $q$, we have
\begin{equation}
\bqR(q;A)-\bqR(q;B)=\bqR(q;A)(\bB-\bA)\bqR(q;B)\,.
\label{res-id}
\end{equation}
This is a consequence of the more general expression for any $\bX$ and $\bY$,
\begin{equation*}
\bX^{-1}-\bY^{-1}=\bX^{-1}(\bY-\bX)\bY^{-1}\,.
\end{equation*}
\item[The resolvent bound.]
For any matrix $A$ and quaternion $q=\lambda+\varepsilon j$, we have the following bound on the norm of the resolvent and its blocks
\begin{equation}
\big\|\bqR(q;A)_{ij}\big\|\leq\big\|\bqR(q;A)\big\|\leq\frac{1}{\varepsilon}\,.
\label{res-bd}
\end{equation}
Again, this is a special case of a more general result - for any matrices $X$ and $Y$, with $Y$ invertible, we have
\begin{equation*}
\big\|(\bX+\bY\bj)^{-1}_{ij}\big\|\leq\big\|(\bX+\bY\bj)^{-1}\big\|\leq \frac{1}{x+y}\,,
\end{equation*}
where $x$ and $y$ are the smallest singular values of $X$ and $Y$.
\item[The Cauchy-Schwarz inequality.] Let $\{\bx_i\}_{i=1}^N$ and $\{\by_i\}_{i=1}^N$ be collections of $2\times 2$ matrices. We have the following analogue of the Cauchy-Schwarz inequality:
\begin{equation}
\left\|\sum_i\bx_i\by_i^\dag \right\|^2\leq \left\|\sum_i\bx_i\bx_i^\dag\right\|\left\|\sum_i\by_i\by_i^\dag\right\|\,.
\label{C-S}
\end{equation}
\item[Integration by parts.] Let $f:\mathbb{C}^{n}\to\mathbb{C}$ be a continuous and differentiable (but not generally analytic) function, with bounded second order partial derivatives 
\begin{equation*}
\sup_{\vec{x},i,j}\left\{\left|\frac{\partial}{\partial x_i\partial x_j }f(\vec{x})\right|,\left|\frac{\partial}{\partial x_i\partial\overline{x_j} }f(\vec{x})\right|,\left|\frac{\partial}{\partial \overline{x_i}\partial\overline{x_j} }f(\vec{x})\right|\right\}<C_f\,.
\end{equation*}
Then if $F=f(\xi_{11},\ldots,\xi_{NN})$, we have
\begin{equation}
\bbE \xi_{ij} F=\bbE \frac{\partial F}{\partial \overline{\xi_{ij}}}+\tau\bbE \frac{\partial F}{\partial \xi_{ji}}+\bbE\xi_{ij}^2\,\bbE\frac{\partial F}{\partial \xi_{ij}}+\bbE\xi_{ij}\overline{\xi_{ji}}\,\bbE \frac{\partial F}{\partial \overline{\xi_{ji}}}+k\,,
\label{parts}
\end{equation}
where $|k|$ is bounded by some constant depending on $C_\xi$ and $C_f$. Proof is by Taylor's Theorem.
\end{description}
In addition to these general facts, the main part of the work in proving Theorems 2 and 3 comes down to the application of two central results.
\begin{lem}
Let $A_N$ be a random $N\times N$ matrix given by (\ref{A}) and let $X$, $Y$ and $Z$ be arbitrary deterministic matrices of the same size, with $Z$ invertible. Define 
\begin{equation*}
\bR=\left(\bX\bA_N+\bY+\bZ\bj\right)^{-1}\,,\quad\textrm{and}\quad\bG=\displaystyle{\frac{1}{N}\sum_i\bR_{ii}}\,.
\end{equation*}
Then 
\begin{equation*}
\bbE\Big\|\bG-\bbE\bG \Big\|^2<C N^{-1}\,,
\end{equation*}
where $C$ is a constant depending on $X,Y$ and $Z$.
\end{lem}
This is simply a slightly more general version of the statement that the quaternionic Green's functions of the matrices we are interested in are self-averaging in the limit $N\to\infty$. This property is crucial if we are to extract useful information about the behaviour in that limit. The proof is based on an elegant martingale technique from a paper on spin-glasses \cite{Carmona2006}.
\begin{proof}
Let $P=\{(i,j):1\leq i\leq j\leq N\}$. We label these pairs by introducing the bijective numbering $(i,j)\leftrightarrow p$ where $p\in\{1,...,|P|\}$. The names $(i,j)$ and $p$ will be used interchangeably; the meaning should be clear from the context. Let $\sF_0=\sF$ and recursively define the sub-$\sigma$-algebras $\sF_p=\sigma\{\sF_{p-1},\xi_{ij},\xi_{ji}\}$. Introduce the martingale 
\begin{equation*}
\Delta_p=\bbE\left(\bG\,\big|\,\sF_p\right)-\bbE\left(\bG\,\big|\,\sF_{p-1}\right)\,,
\end{equation*}
so that
\begin{equation*}
\sum_{p=1}^{|P|}\Delta_p=\bG-\bbE\bG\,.
\end{equation*}
We plan to bound each $\Delta_p$, to do so, we consider the fictitious situation in which the blocks $\big(\bA_N\big)_{ij}$ and $\big(\bA_N\big)_{ji}$ are removed. Write $A_N^{(ij)}$ for the matrix obtained from $A_N$ by setting $\xi_{ij}=\xi_{ji}=0$. Introduce 
\begin{equation*}
\bR^{(ij)}=\left(\bX\bA_N^{(ij)}+\bY+\bZ\bj\right)^{-1}\,,\quad\textrm{and}\quad\bG^{(ij)}=\displaystyle{\frac{1}{N}\sum_k\bR^{(ij)}_{kk}}\,.
\end{equation*}
The resolvent identity (\ref{res-id}) provides
\begin{equation*}
\bR=\bR^{(ij)}-\bR\bX\left(\bA_N-\bA_N^{(ij)}\right)\bR^{(ij)}\,
\end{equation*}
and thus
\begin{equation*}
\bG=\bG^{(ij)}-\frac{1}{N}\bK_{ij}\,,
\end{equation*}
where the error term $\bK_{ij}$ is given by
\begin{equation*}
\begin{split}
\bK_{ij}&=\sum_{k}\left(\bR\bX\left(\bA_N-\bA_N^{(ij)}\right)\bR^{(ij)}\right)_{kk}\\
&=\sum_{k}\left(\big(\bR\bX\big)_{ki}\big(\bA_N\big)_{ij}\bR_{jk}^{(ij)}+\big(\bR\bX\big)_{kj}\big(\bA_N\big)_{ji}\bR_{ik}^{(ij)} \right)\,.
\end{split}
\end{equation*}
The Cauchy-Schwarz inequality (\ref{C-S}) provides a bound for $\bK_{ij}$, since, for example
\begin{equation*}
\begin{split}
\left\|\sum_k\big(\bR\bX\big)_{ki}\big(\bA_N\big)_{ij}\bR_{jk}^{(ij)}\right\|^2&\leq \left\| \sum_k\big(\bR\bX\big)_{ki}\big(\bR\bX\big)^\dag_{ki} \right\|\left\|\big(\bA_N\big)_{ij}\big(\bA_N\big)^\dag_{ij}\right\|\left\| \sum_k\bR^{(ij)}_{jk}\bR^{(ij)\dag}_{jk} \right\|\\
&=\Big\|\left(\bX^T\bR^T\bX^\dag\bR^\dag\right)_{ii}\Big\|\left\|\big(\bA_N\big)_{ij}\big(\bA_N\big)^\dag_{ij}\right\|\Big\|\left(\bR^{(ij)}\bR^{(ij)\dag}\right)_{ii}\Big\|\\
&\leq \frac{CM_{ij}^2}{N}\,,
\end{split}
\end{equation*}
where $M_{ij}=\max\{|\xi_{ij}|,|\xi_{ji}|\}$, and $C$ is the constant coming from the resolvent bound (\ref{res-bd}). We can conclude $\|\bK_{ij}\|<C_KM_{ij}N^{-1/2}$, for some constant $C_K$, and thus 
\begin{equation*}
\begin{split}
\big\|\Delta_p\big\|&=\Big\|\bbE\left(\bG\,\big|\,\sF_p\right)-\bbE\left(\bG\,\big|\,\sF_{p-1}\right)\Big\|\\
&=\Big\|\bbE\left(\bE_{ij}\,\big|\,\sF_p\right)-\bbE\left(\bE_{ij}\,\big|\,\sF_{p-1}\right)\Big\|\\
&\leq C_KN^{-3/2}\big(M_{ij}-\bbE M_{ij}\big)\,.
\end{split}
\end{equation*}
Burkholder's inequality then gives, for some constant $C_{\Delta}$,
\begin{equation*}
\begin{split}
\bbE\Big\|\bG-\bbE\bG \Big\|^3&=\bbE \left\|\sum_{p=1}^{|P|}\Delta_p\right\|^3\leq C_\Delta\,\bbE\left(\sum_{p=1}^{|P|}\big\|\Delta_p\big\|^2\right)^{3/2}\\
&\leq C_\Delta C_K^3N^{-9/2}\,\bbE\left(\sum_{i\leq j}\big(M_{ij}-\bbE M_{ij}\big)^2\right)^{3/2}\\
&<4C_\Delta C_K^3C_MN^{-3/2}\,,
\end{split}
\end{equation*}
where $C_M<2C_{\xi}+2$ is a bound for $\bbE M_{ij}^3$, and we have repeatedly used Jensen's inequality. The desired bound is then given by
\begin{equation*}
\bbE\Big\|\bG-\bbE\bG \Big\|^2<C N^{-1}\,,
\end{equation*}
where $C=(4C_\Delta C_K^3C_M)^{2/3}$.
\end{proof}
With self-averaging established, we next require a mechanism by which we can convert the general statement of the resolvent identity (\ref{res-id}) into an equation for the mean of the quaternionic Green's function.
\begin{lem}
Fix a quaternion $q=\lambda+\varepsilon j$, with $\varepsilon>1$. Let $A_N$ be a random $N\times N$ matrix given by (\ref{A}) and let $X$ be an arbitrary deterministic matrix of the same size. Define 
\begin{equation*}
\bR=\big(\bX\bA_N-\bq\big)^{-1}\quad\textrm{and}\quad\tbG=N^{-1}\sum_i\big(\bR\bX\big)_{ii}\,.
\end{equation*}
Then
\begin{equation*}
\bbE\big(\bA_N\bR\big)=-\big(\bbE (\bt\cdot\tbG)\big)(\bbE\bR)+\bK_N\,.
\end{equation*}
where $t=\tau+j$, and $\|\bK_N\|\to0$ as $N\to\infty$.
\end{lem}
\begin{proof}
We compute a generic block
\begin{equation}
\bbE\big(\bA_N\bR\big)_{ik}=\frac{1}{\sqrt{N}}\bbE\sum_j\bmat\xi_{ij}&0\\0&\overline{\xi_{ji}} \emat_{ij}\bR_{jk}\,.
\label{gen}
\end{equation}
Notice that the entries of $\bR$ depend continuously upon the $\xi$, and the resolvent bound gives a bound on the first and higher order derivatives. We are thus able to apply the integration by parts formula (\ref{parts}). The derivatives are given by
\begin{equation*}
\frac{\partial\bR}{\partial\overline{\xi_{\alpha\beta}}}=-\frac{1}{\sqrt{N}}\bR\bX\bmat 0&0\\0&1 \emat \bm{1}_{\beta\alpha}\bR\quad\textrm{and}\quad \frac{\partial\bR}{\partial\xi_{\alpha\beta}}=-\frac{1}{\sqrt{N}}\bR\bX\bmat 1&0\\0&0 \emat \bm{1}_{\alpha\beta}\bR\,,
\end{equation*}
where $\bm{1}_{\alpha\beta}$ is the $2N\times 2N$ block matrix containing a copy of $I_2$ in block $(\alpha,\beta)$ and zeros elsewhere. Applying this to (\ref{gen}), we thus find, after some tedious algebra, 
\begin{equation*}
\begin{split}
\bbE\big(\bA_N\bR\big)_{ik}&=-\frac{1}{N}\sum_j\left[\left(\bt\cdot\big(\bR\bX\big)_{jj}\right)\bR_{ik}+\left(\bs_{ij}\cdot\big(\bR\bX\big)_{ji}\right)\bR_{jk}\right]\\
&=-\bbE \big((\bt\cdot\tbG)\bR_{ik}\,\big)-\frac{1}{N}\sum_j\left(\bs_{ij}\cdot\big(\bR\bX\big)_{ji}\right)\bR_{jk}\,,
\end{split}
\end{equation*}
where $s_{ij}=\bbE \xi_{ij}^2+(\bbE\xi_{ij}\overline{\xi_{ji}})j$. Notice that assumptions A1-A4 imply a universal bound for $|s_{ij}|$, and the Cauchy-Schwarz inequality together with the resolvent bound give a constant $C$ such that
\begin{equation*}
\frac{1}{N}\left\|\sum_j\left(\bs_{ij}\cdot\big(\bR\bX\big)_{ji}\right)\bR_{jk}\right\|<CN^{-1}\,.
\end{equation*}
Finally, we split the expectation
\begin{equation*}
\left\|\bbE \big((\bt\cdot\tbG)\bR\,\big)-\big(\bbE (\bt\cdot\tbG)\big)(\bbE\bR)\right\|\to0\quad\textrm{as}\quad N\to\infty,
\end{equation*}
since $\bbE\big\|\tbG-\bbE\tbG \big\|^2\to0$ by Lemma 1, $|t|<\infty$ and $\|\bR\|$ is bounded.
\end{proof}

\subsection{Proof of sum and product rules}
\begin{proof}[Proof of Theorem 2]
Fix $q=\lambda+\varepsilon j$ with $\varepsilon>1$. We write the shorthands $\bR_N=\bqR(q;D_N+A_N)$, $\bG_N=\bqG(q;D_N+A_N)$ and $\qG_N(q)=\qG(q;D_N+A_N)$. Now, applying the resolvent identity and Lemma 2, we obtain
\begin{equation*}
\begin{split}
\bbE\bR_N-(\bD_N-\bq)^{-1}&=-(\bD_N-\bq)^{-1}\bbE\bA_N\bR_N\\
&=(\bD_N-\bq)^{-1}\big(\bbE (\bt\cdot\bG_N)\big)(\bbE\bR_N)+\bK_N\,,
\end{split}
\end{equation*}
where $\|\bK_N\|\to0$ as $N\to\infty$. Rearranging, we have
\begin{equation*}
\bbE\bR_N=\Big(\bD_N-\bq-\bt\cdot\bbE\bG_N\Big)^{-1}+\bK_N'\,,
\end{equation*}
where the resolvent bound gives $\|\bK_N'\|\to0$ as $N\to\infty$ also. Summing over the diagonal blocks, we deduce
\begin{equation}
\bbE\qG_N(q)=\qG\Big(q+t\cdot\bbE\qG_N(q);D_N\Big)+k_N\,,
\label{EqGN}
\end{equation}
with $|k_N|\to0$ as $N\to\infty$. Define the functions 
\begin{equation*}
f_N(g)=\qG_{D_N}\Big(q+t\cdot g\Big)+k_N\,,\qquad f(g)=\qG_D\Big(q+t\cdot g\Big)\,.
\end{equation*}
Since $\varepsilon>1$, the resolvent bound gives that each $f_N$ is a contraction with parameter $\varepsilon^{-1}$, and thus the pointwise limit $f$ is also; it is therefore continuous, with a unique fixed point which we call $\qG(q)$. Finally, Lemma 1 gives $\bbE|\qG_N(q)-\qG_N(q)|^2\to0$ as $N\to\infty$, and we conclude from (\ref{EqGN}) and Tchebychev's inequality that 
\begin{equation*}
\qG_N(q)\stackrel{\mathbb{P}}{\longrightarrow}\qG(q)\,.
\end{equation*}
\end{proof}
\begin{proof}[Proof of Theorem 3]
The proof is very similar to that of Theorem 2, so we give only the main steps of the derivation.\par
Fix $q=\lambda+\varepsilon j$ with $\varepsilon>d$. Again we write the shorthands $\bR_N=\bqR(q;D_NA_N)$ and $\qG_N(q)=\qG(q;D_NA_N)$. Let $\tbG_N=N^{-1}\sum_i(\bR_{N}\bD_N)_{ii}$. Applying the resolvent identity and Lemma 2, we obtain
\begin{equation*}
\begin{split}
\bbE\bR_N+\bq^{-1}&=\bq^{-1}\bD_N\bbE\bA_N\bR_N\\
&=-\bq^{-1}\bD_N\big(\bbE (\bt\cdot\tbG_N)\big)(\bbE\bR_N)+\bK_N,
\end{split}
\end{equation*}
Rearranging, we have
\begin{equation*}
\bbE\bR_N=-\big(\bbE (\bt\cdot\tbG_N)\big)^{-1}\Big(\bD_N+\bq\big(\bbE (\bt\cdot\tbG_N)\big)^{-1}\Big)^{-1}+\bK_N'\,,
\end{equation*}
where the resolvent bound gives $\|\bK_N'\|\to0$ as $N\to\infty$ also (it is here that assumption D3 and the requirement $\varepsilon>2d$ are used). Summing over the diagonal blocks, we deduce
\begin{equation}
\bbE\qG_N(q)=-\big(\bbE (t\cdot\tqG_N)\big)^{-1}\qG\Big(-q\big(\bbE (t\cdot\tqG_N)\big)^{-1};D_N\Big)+k_N\,,
\label{EG}
\end{equation}
where $\tqG_N$ is the quaternion with matrix representation $\tbG_N$, and $|k_N|\to0$ as $N\to\infty$. The equation for $\tbG$ is found similarly; let $\tbR_N=\bR_N\bD_N$, then
\begin{equation*}
\begin{split}
\bbE\tbR_N+\bq^{-1}D_N&=\bq^{-1}\bD_N\bbE\bA_N\bR_N\bD_N\\
&=-\bq^{-1}\bD_N\big(\bbE (\bt\cdot\tbG_N)\big)(\bbE\tbR_N)+\tbK_N\,.
\end{split}
\end{equation*}
and thus
\begin{equation}
\bbE\tqG_N=-q^{-1}\qG\left(-\big(t\cdot\bbE\tqG_N)\big)q^{-1};D_N^{-1}\right)+\widetilde{k}_N\,,
\label{EtG}
\end{equation}
where as usual $|\widetilde{k}_N|\to0$ as $N\to\infty$. The proof is completed in the same fashion as Theorem 2, with equations (\ref{EG}) and (\ref{EtG}) providing (\ref{prodtrue}) and (\ref{prod}), respectively.
\end{proof}

\section{Discussion}
The purpose of this paper has been to marry the simple approach to Hermitian RMT \cite{Pastur1999, Khorunzhy1993, Khorunzhy1996, Khorunzhy2001} to the tricks used in \cite{Feinberg1997,Feinberg1997a,Feinberg2001,Janik1997,Janik1997a,Jarosz2004,Jarosz2006} to deal with non-Hermitian matrices, thereby obtaining techniques with which to handle sums and products of random and deterministic matrices. As shown in Theorem 1 the resulting theory in fact applies to the mean spectral density of such matrices under a particular type of random perturbation. In practice it appears, as evidenced by the examples, that Theorems 2 and 3 can indeed be used to predict limiting spectral densities in the absence of this perturbation, though this aspect of the theory has not been rigorously proven. There are two main difficulties in the justification of the exchange of limits $N\to\infty$ and $\varepsilon\to0$ required to make the approach of the examples rigorous.\par
First, in the statement of both results we assume a minimum size for $\varepsilon$ with the purpose of making unique fixed points of equations (\ref{sum}) and (\ref{prod}) easily available. The same trick is used in the analogous theory of Green's functions of Hermitian matrices, however, in that case it is easily justified by appealing to analyticity; one simply determines the Green's function far from the real axis and uses analytic continuation to return. Things are not quite so straightforward in the quaternionic case. A similar argument is still possible if one notes that for fixed $X$, $\qG(\lambda+\varepsilon j;X)$ is a ratio of polynomials in $\varepsilon$ and that the zeros of the denominator are confined to the imaginary axis. One may feasibly then promote $\varepsilon$ to a complex variable on a strip containing the real line and apply analytic continuation as before. The only drawback here is conceptual; preserving the analyticity of $\qG$ in $\varepsilon$ effectively destroys the quaternionic analogy since we must work with $\varepsilon^2$ and not $|\varepsilon|^2$, as would result from the matrix representation (\ref{matrep}).\par
The other problem is more fundamental. For a sequence of matrices $\{X_N\}$, the convergence in probability of $\qG(\lambda+\varepsilon j;X_N)$ translates to the weak convergence of $\varrho_\varepsilon(\lambda;X_N)$ and does not necessarily reveal anything about the limiting behaviour of the unregularised densities $\varrho(\lambda;X_N)$. If $X_N$ happens to be normal, then it is straightforward to compute
\begin{equation*}
\varrho_\varepsilon(\lambda;X_N)=\frac{1}{\pi}\int_{\mathbb{C}}\left(\frac{\varepsilon}{\varepsilon^2+|\mu-\lambda|^2}\right)^2\varrho(\mu;X_N)\,d\mu\,.
\end{equation*}
Since, in this case, the smoothing of $\varrho_\varepsilon$ is independent of $N$, the weak convergence of $\varrho(\lambda;X_N)$ follows easily. Without normality the same cannot be said, in fact it is easy to construct examples for which $\lim_{\varepsilon\to0}\lim_{N\to\infty}\varrho_\varepsilon(\lambda;X_N)$ and $\lim_{N\to\infty}\varrho(\lambda;X_N)$ are entirely different\footnote{Banded Toeplitz matrices are a good choice.}. \par
As mentioned in the introduction, this issue is by no means new or unique to the quaternionic Green's function. Some authors have treated the problem carefully \cite{Bai1997,Pan2007,Gotze2007,Gotze2007a,Tao2007,Chafai2007}, usually by techniques that involve the bounding of the least singular values of the random matrices involved. Such methods may well be adapted to prove the convergence of densities $\varrho(\lambda;D_N+A_N)$ and $\varrho(\lambda;D_NA_N)$ to the limits predicted by the sum and product rules, thereby completing the work of the present paper.\par 
We should point out however, that Theorem 1 suggests the regularised density, for small $\varepsilon$ or not, is an interesting and potentially useful object which is well worthy of study in its own regard. This result also offers a strong heuristic argument for the correctness of the techniques used in the examples - if we are dealing with very large and fully random matrices, the addition of an infinitesimal random perturbation should not change the spectral density.\par
There is some room for improvement in Theorems 2 and 3, both in relaxing the conditions and strengthening the results. For the sake of simplicity and brevity, we assume in A3 the bound $\bbE|\xi_{ij}|^3<C_\xi$, which could almost certainly be dropped in favour of a weaker condition, or possibly forgotten entirely as in \cite{Tao2007}. In a similar vein assumption D3 may be extraneous, as suggested by the example. Lastly, it is possible that the convergence in probability of $\qG$ may be traded up for almost sure convergence, which in turn would provide strong rather than weak convergence of $\varrho_\varepsilon$, however this may well require entirely different methods.\par
The final opportunity for future research worth mentioning is the conjectured `Spherical law', a proof of which, by any method, would be very interesting.

\subsection*{Acknowledgements}
The author would like to thank Isaac P\'{e}rez Castillo for his continued advice and support, and Adriano Barra for important initial discussions.

\newpage
\bibliographystyle{phcpc}
\bibliography{/home/tim/Spectra/Papers/BibTex/Spectra}
\end{document}